\documentclass[11pt]{article}
\usepackage[text={6in,8.5in},centering]{geometry}
\usepackage{mathtools}
\usepackage{multicol}
\usepackage{stmaryrd}
\usepackage{amsmath}
\usepackage{amsthm}
\usepackage{amsfonts}
\usepackage{amssymb}
\usepackage{tikz}
\usepackage[]{algorithm2e}
\usepackage{qcircuit}

\usepackage{hyperref}

\newtheorem{definition}{Definition}
\newtheorem{proposition}{Proposition}

\newtheorem{theorem}{Theorem}
\newtheorem{corollary}{Corollary}
\newtheorem{lemma}[theorem]{Lemma}

\title{Anti-crossings and spectral gap during quantum adiabatic evolution}
\author{Arthur Braida \\ Simon Martiel}

\begin{document}
\maketitle
\begin{abstract}
We aim to give more insights on adiabatic evolution concerning the occurrence of anti-crossings and their link to the spectral minimum gap $\Delta_{min}$. We study in detail adiabatic quantum computation applied to a specific combinatorial problem called weighted max $k$-clique. A clear intuition of the parametrization introduced by V. Choi is given which explains why the characterization isn't general enough. We show that the instantaneous vectors involved in the anti-crossing vary brutally through it making the instantaneous ground-state hard to follow during the evolution. This result leads to a relaxation of the parametrization to be more general.

\end{abstract}

\section{Introduction}
\label{intro}

Adiabatic Quantum Computation (AQC) was introduced by Fahri et al. in \cite{Farhi} as an alternative to the standard digital approach of quantum computing. In this setting, the quantum program is described as a time-dependent Hamiltonian driving some continuous-time evolution of a quantum system. The evolution of the quantum system is then given by Schrödinger's Equation:

\begin{align}
i \frac{ \partial}{\partial t}|\psi(t)\rangle=H(t) |\psi(t)\rangle
\end{align}
where the Planck's constant $\bar{h}$ has taken the value of unity and $H(t)$ is the time dependant Hamiltonian that rules the evolution of $|\psi (t) \rangle$. 
Usually, in the literature, this equation is expressed using the adimensional time $s=\frac{t}{T}$ as:
\begin{align}
i \frac{\partial}{\partial s}|\psi(s)\rangle=TH(s) |\psi(s)\rangle
\end{align}
thus removing the dependency in the execution time $T$.
In general, in AQC, an ordinary setting to solve hard combinatorials problems with a unique solution is to describe the time-dependant Hamiltonian $H(s)$ as an interpolation between an initial Hamiltonian $H_0$ easy to capture and a final Hamiltonian $H_1$ in which we encode the problem:
\begin{align}
\forall s \in [0,1],  H(s) = f(s)H_0+g(s)H_1
\end{align}
where $f(0)=g(1)=1$ and $f(1)=g(0)=0$. 
For instance, it is trivial to encode Quadratic Unconstrained Binary Optimization (QUBO) problems using Ising Hamiltonians \cite{10.3389/fphy.2014.00005,Farhi}. A recent overview in 2016 \cite{Lidar} shows the different applications and results in this framework.
\\ \\
In the AQC framework, we are interested in finding the optimal solution to our problem which corresponds to the ground-state of our final Hamiltonian $H_1$. Let's write the unique optimal solution of the given problem as $|GS\rangle$ i.e. $H(1)|GS\rangle = E_0(1)|GS\rangle$. We want to maximise the overlap of  $|GS\rangle$ with $|\psi(1)\rangle$ i.e. the probability $|\langle GS|\psi(1)\rangle|^2$ . By the adiabatic theorem, if the system evolves slowly enough and starts in the ground-state of $H_0$, it will stay in the ground-state during the evolution and will end in the ground-state of $H_1$, ensuring an overlap close to one. Therefore to take advantage of the adiabatic regime, the quantum system must be initiated as the ground-state of $H_0$, i.e. $H_0|\psi(s=0) \rangle=E_0(0)|\psi(s=0) \rangle$. The adiabatic theorem characterizes \textit{slowly enough} by giving a lower bound on the execution time T:
$$
T \gg \frac{\epsilon}{\Delta_{min}^2}
$$
where $\epsilon = \max_s | \langle E_1(s) | \frac{dH}{ds} |E_0(s) \rangle |$ and $\Delta_{min}=\min_s E_1(s) - E_0(s)$ where we introduce the eigen decomposition of $H(s)$: $H(s)| E_k(s)\rangle=E_k(s) |E_k(s)\rangle $ with $E_0(s)<E_1(s)\leq...\leq E_{2^n}(s)$.  Notice that only the ground-state needs to be non-degenerated, $|E_0(0)\rangle=|\psi(0) \rangle$ is the initial state and $|E_0(1)\rangle = |GS\rangle$ is the targeted state toward which the state $|\psi(s)\rangle$ evolves. In typical problem of interests, with reasonable interpolation's functions, $\epsilon$ scales polynomially with the problem size, so the run time is mostly ruled by the minimum gap. Evolutions with exponentially small minimum gap will lead to exponentially large annealing times. Usually, we call $H_0$ the mixing Hamiltonian because it allows the algorithm to search for solution in the Hilbert space and $H_1$ is the target Hamiltonian which is diagonal in the computational basis and encodes a classical cost function.
\\ \\
To solve a combinatorial problem with an adiabatic evolution, the main quantity to measure the efficiency of the algorithm is the min-gap of the given problem. A first approach to conclude on the efficiency of the algorithm is to distinguish between exponentially small min-gap or not. A useful tool for that matter is the analysis of physical phenomena happening during AQC called anti-crossings. It corresponds to specific behavior of the two lowest eigenvalues when the min-gap occurs. They were first described by Wilkinson \cite{Wilkinson1989} as two branches of hyperbola around the anti-crossing point $s^*$ with the following expressions: \begin{align}\label{wilkAC}
E^{\pm}(s) = E(s^*) + B(s -s^*) \pm \frac{1}{2}[\Delta_{min}^2 + A^2(s-s^*)^2]^{1/2}
\end{align} where $A$ and $B$ are respectively the difference and the mean of the slopes of the asymptotes of the hyperbola. The presence of an anti-crossing at a first-order phase transition during the evolution leads to an exponentially small min-gap. Thus to show that a specific adiabatic evolution will fail, it is enough to exhibit the occurrence of one anti-crossing. One method to analyze the occurrence of this phenomena at the end of the evolution is the perturbative framework where one tries to write the eigenvalues as a perturbative expansion. This has been done in \cite{Altshuler2018} by Altshuler et al. to show that AQC fails for solving NP-Hard problems even in the average setting for a given time-dependent Hamiltonian. Anti-crossings can occur at any time of evolution, so not all anti-crossing can be found using perturbative theory. This is why the characterization of the phenomena for AQC is still an active research subject. Recently, Vicky Choi suggested a new parametrization of the physical phenomenon to study the occurrence of anti-crossings during AQC \cite{Choi2019}. However, the link between the anti-crossing and the min-gap is not clear in the literature and as far as we know, it seems to be an admitted result. A second approach is to derive an analytical expression of the gap and directly find its minimum. Sun and Lu published in 2019 \cite{Sun2019} some results on a general criterion to establish the failure of AQC. They derived some analytical expressions of the min-gap after projecting the eigen relation on the subspaces span by the initial ground-state and the final ground-state. We suggest another expression of the min-gap that is more general. 
\\ \\
In this work, we show a new general expression for the min-gap based on a spectral development. Then we study the relationship between the definition of an anti-crossing given by \cite{Choi2019} and the min-gap. We introduce the problem of weighted $k$-clique to analyze in detail the definition and explain why it's not always satisfied. Finally, we show how the eigenvectors involved in the anti-crossing vary through it which leads us to our new characterization, more general, of an anti-crossing. The paper is organized as follows. In section 2, like Sun and Lu in \cite{Sun2019}, we derive a more general expression of the min-gap and a necessary condition for the failing of AQC. In section 3, we study the gap by analyzing the phenomenon of anti-crossing which happens during an adiabatic evolution. Different frameworks are suggested in the literature to study them, the last we've seen is the one described by V. Choi in 2020. We aim to give more insights of her parametrization and directly link it to the min-gap which for us was a missing result. Then we explain why the characterization doesn't capture an anti-crossing in general by illustrating it with the problem of finding weighted $k$-clique in random graphs and show that it doesn't work in general. In section 4, we offer a new characterization that captures more instances and we prove that our parametrization can effectively quantify the strength of the anti-crossing. At the end, we suggest new leads to pursue the understanding of the behaviors of these quantities during the evolution. In the appendices, we detail an example to explain and give full intuition in the comprehension of the parametrization (from V. Choi and the one presented in this paper). 

\section{Gap general expression}

In this section, we introduce a new notation and we derive a general expression of the min-gap by projecting the eigen relation on the computational basis and the instantaneous basis. We first introduce a way to visualize the mixing term by looking at the graph with adjacency matrix $-H_0$. We call this graph $G(H_0)$. For example, with a classical transverse field $H_0 =-\sum_i^n \sigma_x^{(i)}$ we get the hypercube in dimension $n$ (Figure \ref{ghx}). In this situation, with only bit-flip moves allowed, one can move from one node to another by flipping one bit from the bitstring (applying the $\sigma_x$ operator to the bit that is flipped). $G(H_0)$ represents the exploration state space of the algorithm. It shows how the different states communicate which each other and how we can navigate from one state to the other. \\
\begin{figure}[ht]
\centering
\includegraphics[scale=0.4]{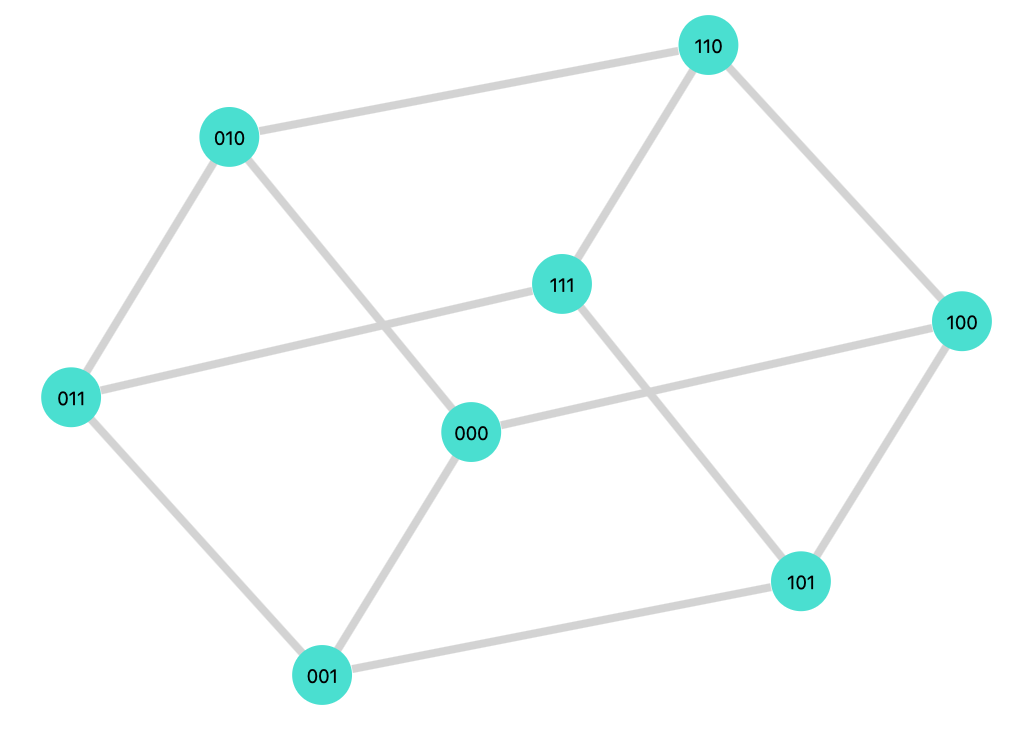}
\caption{$G(H_0)$ with the transverse field $n=3$}
\label{ghx}
\end{figure}

\noindent Each state represented on the graph is an eigen vector of the final Hamiltonian $H_1$ and all of them form the computational basis that will be denoted as $|\mathbf{x}_i \rangle$ ($|\mathbf{x}_i \rangle = |E_i(1)\rangle$). They all have a specific final energy $E_i(1)$ and the solution is the one with the lowest energy $E_0(1)$, i.e. the ground-state  $|E_0(1)\rangle = |GS\rangle$. \\

\noindent The quantity of interest is the behavior of the gap $\Delta(s)=E_1(s)-E_0(s)$ during the evolution. It is known that in general, we cannot derive analytical expression. Here, we express this gap in a similar way that was done in \cite{Sun2019}. Let's first isolate the two eigenvalues. By definition of the eigen decomposition we have 
$$
H(s)|E_k(s) \rangle = E_k(s) |E_k(s) \rangle
$$we compose on the left by an computational basis vector $|\mathbf{x}_i \rangle$ and we obtain 
$$
\langle \mathbf{x}_i |(f(s)H_0+g(s)H_1)|E_k(s) \rangle = E_k(s) \langle \mathbf{x}_i |E_k(s) \rangle 
$$because $H_1 |\mathbf{x}_i\rangle = E_i(1) |\mathbf{x}_i\rangle$, we end up with 
$$
f(s)\langle \mathbf{x}_i|H_0|E_k(s)\rangle + g(s) E_i(1) \langle \mathbf{x}_i |E_k(s) \rangle= E_k(s) \langle \mathbf{x}_i |E_k(s) \rangle .
$$Now we can isolate $E_k(s)$ to get
\begin{align}
\label{eval}
E_k(s) = g(s) E_i(1) - f(s)\frac{\langle \mathbf{x}_i|(-H_0)|E_k(s)\rangle}{\langle \mathbf{x}_i |E_k(s) \rangle} .
\end{align}
\noindent Note that he vector $(-H_0)|\mathbf{x}_i\rangle$ is the state that represents the neighbors of $|\mathbf{x}_i\rangle$ in $G(H_0)$, therefore we can write it like $|neigh_{H_0}(\mathbf{x}_i)\rangle$. This result links the evolution of the $i^{th}$ component of the $k^{th}$ instantaneous eigen vector to the component corresponding to the neighbors of the state $|\mathbf{x}_i\rangle$. It expresses how each state influences its neighbors or inversely how each state is influenced by its neighbors. We can now apply (\ref{eval}) with $k=0,1$ and formulate the gap $\Delta(s)$ as :
\begin{align}
\Delta(s) = f(s) \left( \frac{\langle neigh_{H_0}(\mathbf{x}_i) |E_0(s)\rangle}{\langle \mathbf{x}_i|E_0(s)\rangle} - \frac{\langle neigh_{H_0}(\mathbf{x}_i) |E_1(s)\rangle}{\langle \mathbf{x}_i|E_1(s)\rangle} \right)
\end{align}This is the most general expression of the gap presented here as it is true for any $i$. If we look at these expressions with $|\mathbf{x}_i \rangle = |GS \rangle$, we can understand the terms in parenthesis as the comparison of the presence of the solution in the instantaneous ground state $|\langle GS|E_0(s)\rangle|^2$ and the presence of its neighbors. The same holds for the first excited instantaneous state $|E_1(s) \rangle$. On one hand, the run time will be polynomial if the difference is only polynomially small. On the other hand, if it exists $s^*<1$ for which we have

\begin{align}
\frac{\langle neigh_{H_0}(GS) |E_0(s^*)\rangle}{\langle GS|E_0(s^*)\rangle} \simeq \frac{\langle neigh_{H_0}(GS) |E_1(s^*)\rangle}{\langle GS|E_1(s^*)\rangle}
\end{align}where we get rid of the interpolation coefficient $f(s)$ and $g(s)$, AQC will be a failure as the runtime will be exponentially large. This condition generalizes the expression found by Sun and Lu in \cite{Sun2019}. One can also write the following bounds on the min-gap occurring at $s^*<1$ by using triangular inequality: 

\begin{itemize} 
\item[$\bullet$]  $\Delta_{min}^2 \leq f(s^*)^2 \left( \frac{ |\langle neigh_{H_0}(\mathbf{x}_i) |E_0(s^*)\rangle|^2}{|\langle \mathbf{x}_i|E_0(s^*)\rangle|^2} +\frac{|\langle neigh_{H_0}(\mathbf{x}_i) |E_1(s^*)\rangle|^2}{|\langle \mathbf{x}_i|E_1(s^*)\rangle|^2} \right)$
\item[$\bullet$]  $\Delta_{min}^2 \geq f(s^*)^2 \left( \frac{ |\langle neigh_{H_0}(\mathbf{x}_i) |E_0(s^*)\rangle|^2}{|\langle \mathbf{x}_i|E_0(s^*)\rangle|^2} - \frac{|\langle neigh_{H_0}(\mathbf{x}_i) |E_1(s^*)\rangle|^2}{|\langle \mathbf{x}_i|E_1(s^*)\rangle|^2} \right)$
\end{itemize}
In \cite{Choi2019}, V. Choi introduces the concept of Low Energy Neighboring States (LENS) to explain in which situation the min-gap will be large or small. Looking at $G(H_0)$ centered on the solution $|GS\rangle$, the min-gap will be larger if the solution is surrounded by states with low energy. With the derived bounds presented above, we understand how the state $|neigh_{H_0}(GS)\rangle$ can affect the min-gap. 
\\ \\
This approach to measure the efficiency of the adiabatic evolution is rather still hard to apply because the existence of this $s^*$ isn't easy to show. The main point when talking about algorithms is to distinguish between exponentially large runtime or not. In the next section, we analyze what can create an exponentially small min-gap during the evolution.

\section{Anti-crossings and gap}
In this section, we will focus on anti-crossings which is known as a physical phenomenon happening in AQC that brings an exponentially small gap. We will study the more recent parametrization of this presented by V. Choi in \cite{Choi2019} and we will explicitly link her parameters with the gap which seems to be an admitted result. Then, with toy-example of the problem weighted max clique, we will show that her definition isn't always satisfied. 
\\ \\
From now on, we restrict the time dependant Hamiltonian to a linear interpolation, namely: $f(s)=1-s$ and $g(s)=s$. Let's state a useful lemma that applies in this setting:

\begin{lemma}
\label{eigenrelation}
Let's $\lambda_i, |n_i\rangle$ be an eigenvalue/eigenvector orthogonal pair of a real-valued symmetric matrix operator $B$ such that $\ddot{B}=0$, we have:
\begin{align}
\label{firstder}
&\dot{\lambda_i} = \langle n_i |\dot{B}|n_i \rangle \\
\label{firstdervec}
&|\dot{n_i}\rangle = \sum_{j \neq i} \frac{\langle n_j |\dot{B}|n_i \rangle}{\lambda_i - \lambda_j} |n_j\rangle \\
&\ddot{\lambda_i} = 2\sum_{j \neq i} \frac{|\langle n_j |\dot{B}|n_i \rangle |^2}{\lambda_i - \lambda_j}  
\end{align}
\end{lemma}
 \begin{proof}
 We take the derivative of the eigen relation $B|n_i\rangle = \lambda_i |n_i\rangle$:
 $$
 \dot{B} |n_i \rangle + B |\dot{n_i} \rangle = \dot{\lambda_i}|n_i\rangle + \lambda_i |\dot{n_i} \rangle
 $$and knowing that $\langle n_i | \dot{n_i} \rangle=0$ because $\langle n_i |n_i\rangle=1$, we compose by $\langle n_i |$ on the left to get the first expression. Then, composing the same expression by another eigenvector $|n_j\rangle$ with eigenvalue $\lambda_j$ we get $(\lambda_i - \lambda_j) \langle n_j|\dot{n_i} \rangle = \langle n_j |\dot{B}|n_i \rangle$. From that we get the second expression. Eventually, we take the second derivative of the eigen relation:
 
 $$
  \ddot{B} |n_i \rangle + 2\dot{B} |\dot{n_i} \rangle +B |\ddot{n_i} \rangle= \ddot{\lambda_i}|n_i\rangle + 2\dot{\lambda_i} |\dot{n_i} \rangle + \lambda_i |\ddot{n_i} \rangle
 $$by hypothesis $\ddot{B}=0$, then projecting on $ |n_i \rangle$ and using (\ref{firstdervec}) give the third result.
 \qed
 \end{proof}
\subsection{Anti-crossings in AQC}

As mentioned previously, the complexity of an algorithm in AQC is in general very difficult to study analytically. Few examples are known where the expression of the gap can be found explicitly \cite{Lidar}. Therefore, to study the efficiency of AQC for a specific task, one way to solve it is by proving the occurrence of a strong anti-crossing at a first-order phase transition. Indeed, it is admitted that a first-order phase transition brings an exponentially small gap. In \cite{Altshuler2018}, Altshuler et al. studied the anti-crossings in a perturbative framework because it allows them to use the perturbative theory and the tools to analyze the eigenvalues of the perturbed Hamiltonian $\bar{H}(s) = H_1 + \lambda H_0$ with a small $\lambda$. The eigenvalues can be expressed as a perturbative expansion in $\lambda$ : $E_\mathbf{x}(\lambda) = E_\mathbf{x}+\sum_{q=1}^\infty \lambda^q E_\mathbf{x}^{(q)}$ and the coefficients $E_x^{(q)}$ can be derived analytically. By studying the behavior of these factors, one can conclude on the presence or not of a strong anti-crossing. However, this method works only when you can actually find anti-crossings at the end of the adiabatic evolution. From \cite{Laumann}, we know that they can happen at any time.  
\\ \\
In \cite{Choi2019}, the author suggests a new parametrization of an anti-crossing during an adiabatic evolution. From the states $|E_k(1)\rangle$ which can be degenerated and represents the states at the $k^{th}$ energy level, V. Choi introduces the following quantities:
\begin{align}
a_k(s)=|\langle E_k(1) |E_0(s) \rangle |^2 \\
b_k(s)=|\langle E_k(1) |E_1(s) \rangle |^2 
\end{align}
They are the decomposition in the possible degenerated computational basis of the instantaneous eigenvectors corresponding to the two lowest eigenvalues. Especially, with $k=0,1$, we have $a_0(s)$ and $a_1(s)$ representing respectively how much of the final ground-state $|E_0(1)\rangle = |GS\rangle$ and the final first excited state $|E_1(1)\rangle = |FS\rangle$ overlaps with the instantaneous ground-state $|E_0(s) \rangle$, i.e. $a_0(s)$ is the probability of finding $|GS\rangle$ in $|E_0(s)\rangle$. The same is true for $b_0(s)$ and $b_1(s)$ with the instantaneous first excited state $|E_1(s) \rangle$. Let's restate her definition of a $(\gamma,\epsilon)$-anti-crossing:

\begin{definition}
\label{ChoiAC}
For $\gamma \geq 0$, $\epsilon \geq 0$ we say there is an $(\gamma, \epsilon)$-Anti-crossing if there exists a $\delta>0$ such that
\begin{enumerate}
\item For $s\in[s^*-\delta,s^*+\delta]$, \begin{align}
|E_0(s) \rangle &= a_0(s) |GS \rangle + a_1(s) |FS\rangle \\
|E_1(s) \rangle &= b_0(s) |GS \rangle - b_1(s) |FS \rangle
\end{align} where $a_0(s)+a_1(s) \in [1-\gamma,1]$, $b_0(s)+b_1(s) \in [1-\gamma,1]$. Within the time interval $[s^*-\delta, s^*+\delta]$, both $|E_0(s)\rangle$ and $E_1(s)\rangle$ are mainly composed of $|GS\rangle$ and $|FS\rangle$. That is, all other states (eigenstates of the problem Hamiltonian $H_1$) are negligible (which sums up to at most $\gamma \geq 0$).

\item  At the avoided crossing point $s=s^*$, $a_0$, $a_1$, $b_0$, $b_1 \in [1/2 - \epsilon,1/2+\epsilon]$, for a small $\epsilon >0$. That is, $|E_0(s^*)\rangle \simeq 1/\sqrt{2}(|GS\rangle + |FS\rangle)$ and $|E_1(s^*)\rangle \simeq 1/\sqrt{2}(|GS\rangle - |FS\rangle)$.

\item  Within the time interval $[s^*-\delta, s^*+\delta]$, $a_0(s)$ increases from $\leq \gamma$ to $\geq (1-\gamma)$, while $a_1(s)$ decreases from $\geq (1-\gamma)$ to  $\leq \gamma$. The reverse is true for $b_0(s)$, $b_1(s)$. 
\end{enumerate}

\end{definition} 

\noindent This definition of an anti-crossing gives more insights toward the understanding of this physical phenomenon happening during an adiabatic evolution. Four new quantities, $a_0(s)$,$a_1(s)$ and $b_0(s)$,$b1(s)$, and two parameters $(\gamma,\epsilon)$ are at stake here, to describe it. Definition \ref{ChoiAC} presents in a precise way how the quantities vary through the anti-crossing and implicitly suggests on the size of the parameters that the smaller they are, the stronger the anti-crossing will be. However, there is a missing result that directly links this definition of an anti-crossing to the min-gap. How the parameters $(\gamma, \epsilon)$ influence the min-gap?  We derive here the exact expression (Proposition \ref{excatMG}) of the min-gap using the quantities introduced by V. Choi.

\begin{proposition}
\label{excatMG}
$\Delta_{min} = \sum_k E_k(1) \left [b_k(s^*) - a_k(s^*) \right ]$
\end{proposition}

\begin{proof}
The gap reaches a minimum at $s^*$ so its derivative is null at $s^*$, thus using the first formula of lemme \ref{eigenrelation}, we get:

\begin{align*}
0 &=\frac{d\Delta}{ds}(s^*)\\
   &=\frac{dE_1}{ds}(s^*) - \frac{dE_0}{ds}(s^*) \\
   &= \langle E_1(s^*) |\dot{H} |E_1(s^*)\rangle - \langle E_0(s^*) |\dot{H} |E_0(s^*)\rangle 
\end{align*}
\begin{multline*}
  \Rightarrow \langle E_1(s^*) |H_0 |E_1(s^*)\rangle - \langle E_0(s^*) |H_0|E_0(s^*)\rangle=\langle E_1(s^*) |H_1 |E_1(s^*)\rangle \\
  - \langle E_0(s^*) |H_1|E_0(s^*)\rangle
\end{multline*}
In the setting of linear interpolation, we have $\dot{H}=H_1-H_0$, so we can rewrite the eigen relation like $s\dot{H}|E_k(s)\rangle = E_k(s)|E_k(s)\rangle - H_0|E_k(s)\rangle$. Then projecting against $\langle E_k(s)|$ gives $s\langle E_k(s)|\dot{H}|E_k(s)\rangle=E_k(s) - \langle E_k(s)|H_0|E_k(s)\rangle$. Applying this relation with $k=0,1$ on the previous calculation, gives another expression of the min-gap:
\begin{align*}
\Delta_{min}&=E_1(s^*) - E_0(s^*) \\
 						&=\langle E_1(s^*)|H_0|E_1(s^*)\rangle-\langle E_0(s^*)|H_0|E_0(s^*)\rangle \\
                        &=\langle E_1(s^*)|H_1|E_1(s^*)\rangle-\langle E_0(s^*)|H_1|E_0(s^*)\rangle
\end{align*}Now, let's write $|E_1(s^*)\rangle$ and $E_0(s^*)\rangle$ in the computational basis $|\mathbf{x}_i\rangle$. There exists $\alpha_i$ and $\beta_i$ such that 
\begin{align*}
|E_0(s^*)\rangle =\sum_i \alpha_i |\mathbf{x}_i\rangle \quad \text{and} \quad
|E_1(s^*)\rangle =\sum_i \beta_i |\mathbf{x}_i\rangle
\end{align*}By definition of $a_k$ and $b_k$, we have that 
\begin{align*}
\sum_{i\text{ s.t.} |\mathbf{x}_i \rangle \in \text{span}(|E_k(1)\rangle)}\alpha_i^2 = a_k(s^*) \quad \text{and} \quad
\sum_{i\text{ s.t.}|\mathbf{x}_i \rangle \in \text{span}(|E_k(1)\rangle)}\beta_i ^2= b_k(s^*)
\end{align*}Then  knowing that $H_1|\mathbf{x}_i\rangle = E_i(1) |\mathbf{x}_i\rangle$ and $E_i(1)$ is the same for all the subspace span by $ |E_k(1)\rangle$, we have: 
\begin{align*}
\langle E_0(s^*)|H_1|E_0(s^*)\rangle  &= \sum_j \sum_i\alpha_j \alpha_i \langle \mathbf{x}_j |H_1|\mathbf{x}_i\rangle \\
&=\sum_i \alpha_i^2 E_i(1) \\
&=\sum_k a_k(s^*) E_k(1)
\end{align*}The same goes with $|E_1(s^*) \rangle$ and $b_k(s^*)$ therefore we end up with:
$$
\Delta_{min} = \sum_k E_k(1) \left [b_k(s^*) - a_k(s^*) \right ]
$$\qed
\end{proof}
\noindent This result is a general expression of the min-gap when linear interpolation is used. We can see that it depends only on the new variables and the final energies. The min-gap will be exponentially small if each terms $a_k$ and $b_k$ are small or if $a_k \simeq b_k$. We can now upper-bound the min-gap using the parameter $\epsilon$ that quantifies the strength of the anti-crossing.
\begin{corollary}
For a $(\gamma, \epsilon)$-anti-crossing of definition \ref{ChoiAC}, we have:
$$
\Delta_{min} \leq K\epsilon
$$ for some constant $K$.
\end{corollary}
\begin{proof}
For a $(\gamma, \epsilon)$-anti-crossing of definition \ref{ChoiAC}, $a_0(s^*), a_1(s^*), b_0(s^*), b_1(s^*) \in [1/2 - \epsilon, 1/2 + \epsilon]$. Thus, for $i=0,1$, $|b_i(s^*) - a_i(s^*)| \leq 2\epsilon$. This also means that $\sum_{i\geq2}a_i(s^*) \leq 2\epsilon$, and the same goes for the $b_i(s^*)'$s. The final energy $E_k(1)$ is upper-bounded by a constant $M$ and without loss of generality we assume $E_k(1) \geq 0$, therefore:
\begin{align*}
\Delta_{min} &\leq E_0(1) |b_0(s^*)-a_0(s^*)| + E_1(1) |b_1(s^*)-a_1(s^*)|\\
& \qquad \qquad \qquad \qquad \qquad \qquad \qquad \qquad \qquad \qquad +M\sum_{k\geq2}[b_k(s^*)+a_k(s^*)]\\
& \leq 2\epsilon E_0(1) + 2\epsilon E_1(1) + 4\epsilon M \\
& \leq 2(E_0(1)+E_1(1) +2M)\epsilon
\end{align*}\qed
\end{proof}

\noindent The smaller $\epsilon$ is, the smaller the min-gap will be. Therefore, if the instantaneous ground-state is purely a linear combination of $|GS\rangle$ and $|FS\rangle$ (i.e. $\epsilon$ is exponentially small), the min-gap will be exponentially small. This new result quantifies the strength of a V. Choi anti-crossing. 
\\ \\
\textbf{Intuition}: $a_k$'s can be seen as the direction toward which the instantaneous ground-state is evolving. For a particular $j$, if $a_j$ is becoming dominant at some point in the evolution, it means that $|E_0(s)\rangle$ is going toward the $j^{th}$ final energy.  We understand that the definition \ref{ChoiAC} isn't general enough as it is only focused on $a_0$ crossing with $a_1$, but in theory, $a_0$ could cross with any other level $a_j$. The same intuition holds for the $b_k$'s and the instantaneous first excited state $|E_1(s)\rangle$. In the appendices, we give a detailed example based on the problem described in \ref{MWkC} to explain each variable with different situations.
\\ \\
Here, we presented the parametrization suggested by V. Choi of an anti-crossing and completed it with a new result that links her definition with the min-gap. 
In the following subsection, we introduce a toy problem in order to illustrate this intuition and construct counter-examples to Choi's definition of anti-crossings.

\subsection{Maximum-Weight $k$-clique problem \label{MWkC}}

The Maximum-Weight $k$-clique problem as an optimization version is defined as finding a $k$-clique maximizing the weights of the nodes in a random graph. The random graph will be denoted by $G(E,V)$ where $E$ is the set of edges and $V$ the set of nodes labelled from $1$ to $n$ with each having a weight $w_i$. In \cite{Childs2000}, the authors gave a way to encode in a target Hamiltonian the solution of finding a $k$-clique (without taking the weight into consideration first). In the latter problem, we are only interested in sub-graphs of size $k$, thus we can restrict the whole Hilbert space of size $2^n$ to the $\binom{n}{k}$ Hilbert space, that is the Hilbert space spanned by bit-strings of Hamming weight $k$. Consequently, the mixing Hamiltonian $H_0$ must stabilize this Hilbert space and thus preserve the Hamming weight of the computational basis vectors encoding our subgraphs. The target Hamiltonian $H_1$ adds an energy penalty for each edge missing in the subgraph. A natural choice for $H_0$ is:

$$
 H_0 = - \sum_i S^{i,i+1}
$$
where
$$ S = \left(\begin{matrix}
0 & 0 & 0 & 0\\
0 & 0 & 1 & 0\\
0 & 1 & 0 & 0\\
0 & 0 & 0 & 0
\end{matrix}\right) ^{ij}
\text{ swaps qubit $i$ and $j$ } $$ 
and for the final Hamiltonian:
 $$ H_1 = \sum_{i, j \notin E} A^{i j} $$
where 
$$ A = \left(\begin{matrix}
0 & 0 & 0 & 0\\
0 & 0 & 0 & 0\\
0 & 0 & 0 & 0\\
0 & 0 & 0 & 1
\end{matrix}\right), \text{ is a "logical AND" between two qubits.} $$
Then, knowing that $ 4A = Z \otimes Z + I \otimes I - Z \otimes I - I \otimes Z $, and $2S = X \otimes X + Y \otimes Y$, where $X$, $Y$ and $Z$ are the Pauli matrices, we can restate the Hamiltonians using Pauli operators:
\begin{align*}
H_0 =& -\frac{1}{2} \sum_i (\sigma_x^{(i)} \sigma_x^{(i+1)} + \sigma_y^{(i)} \sigma_y^{(i+1)} ) \\
H_1 =& \frac{1}{4} |\bar{E}| \mathbb{I} + \frac{1}{8} \sum_{i,j} \bar{G_{ij}} \sigma_z^{(i)} \sigma_z^{(j)} -  \frac{1}{4} \sum_i \bar{\deg{(i)}} \sigma_z^{(i)}
\end{align*}
where $\bar{G}$, $\bar{E}$ and $\bar{deg()}$ is respectively the adjacency matrix, the edges set and the degree function of the complementary graph of $G$. Then, we extract the states of Hamming weight $k$. Now, to take into account the weighted version of the problem, we have to add a term in the final Hamiltonian that favors the total weight of a sub-graph, namely $ -\alpha\sum_i \frac{1-\sigma_z^{(i)}}{2}w_i$ where $\alpha$ is a parameter we can play with to give more or less importance to the weights and $w_i$ is the weight of the node labelled $i$. \\

\noindent The toy example we constructed to illustrate anti-crossing phenomenon is a graph on $n=6$ vertices and $|E|=7$ edges and the size of the clique we search is of size $k=3$ (see figure \ref{toy1}). Let's $w=[1.0,1.0,1.0,1.5,1.5,1.5]$ be the list of weights of the six nodes. Using this weights vector, for $\alpha < 2/3$, the ground state of $H_1$ is $|GS \rangle =|111000\rangle$ with energy $0-3\alpha$ and the first exited state $|FS\rangle = |000111 \rangle$ with energy $1-4.5\alpha$ (for $\alpha=0$, there are many first excited states). The ground-state is degenerated for 
$
1-4.5\alpha = -3\alpha
$ i.e. $\alpha=2/3$. \\

\begin{figure}[ht]
  \centering
  \begin{tikzpicture}
  \node[circle, draw, inner sep=1pt](v6) at (0,0){6};

  \node[circle, draw, inner sep=1pt](v1) at (2,0){1};
  \node[circle, draw, inner sep=1pt](v3) at (3,-1){3};
  \node[circle, draw, inner sep=1pt](v4) at (-1,-1){4};

  \node[circle, draw, inner sep=1pt](v5) at (0,-2){5};
  \node[circle, draw, inner sep=1pt](v2) at (2,-2){2};

  \draw (v1) -- (v3) -- (v2) -- (v1) -- (v6) --(v5) -- (v2);
  \draw (v6) -- (v4) ;

  \end{tikzpicture} 
  \caption{Toy example 1 with weights [1,1,1,1.5,1.5,1.5] for each node. For $\alpha < 2/3$, the solution is the triangle with nodes labelled 1, 2 and 3. For $\alpha>2/3$, the solution is the subgraph \{4,5,6\}.}
  \label{toy1}
\end{figure}
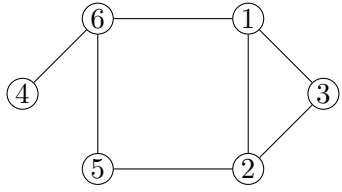

\noindent Now let's have a look at the evolution of the different quantities occurring during an anti-crossing to compare two instances, $\alpha=0$ and $\alpha=0.5$. First, let's explicit the final energies with the different states associated (figure \ref{States_toy1}). We see that in the case of $\alpha=0$, the first excited state is degenerated with 8 states whereas, with $\alpha=0.5$, both ground and first excited states are non-degenerated. \\

\begin{figure}[ht]
\centering
$$
\begin{array}{cc}
\includegraphics[scale=0.4]{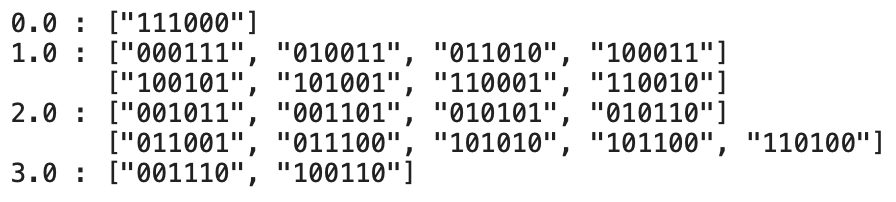} &
\includegraphics[scale=0.4]{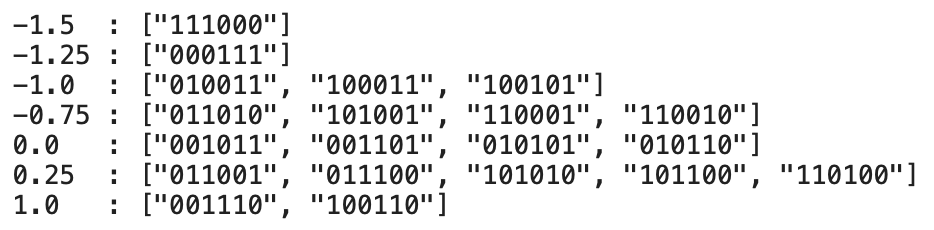} \\
(a) & (b)
\end{array}
$$
\caption{States with energies for toy example 1 for $\alpha=0$ (a) and $\alpha=0.5$ (b)}
\label{States_toy1}
\end{figure}

\begin{figure}
\centering
$$
\begin{array}{cc}
\includegraphics[scale=0.3]{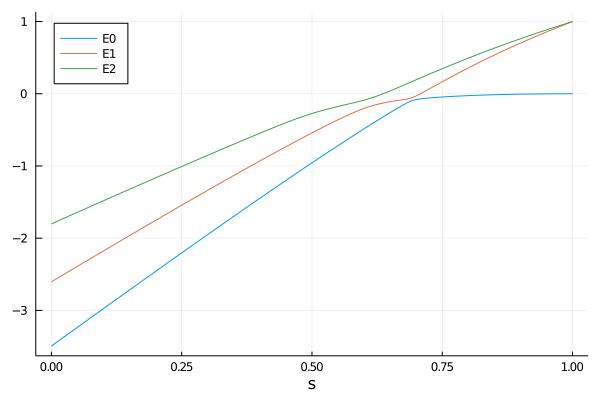} &
\includegraphics[scale=0.3]{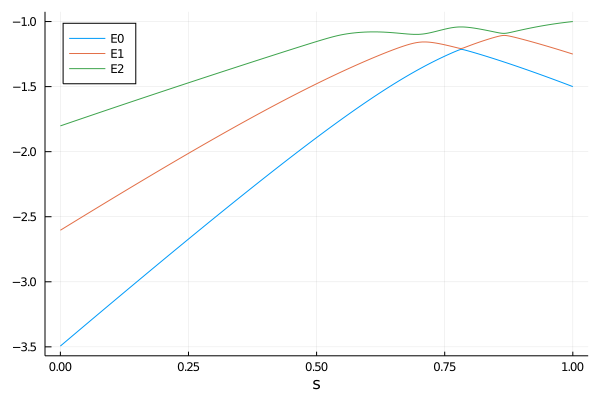} \\
(a) & (b)
\end{array}
$$
\caption{Eigenvalues evolution for toy example 1 with $\alpha=0$ (a) and $\alpha=0.5$ (b)}
\label{Ei_toy1}
\end{figure}

\noindent In the following analysis, we call anti-crossing the two hyperbolas described by Wilkinson (\ref{wilkAC}) that can happen between any two eigenvalues and we denote by $s^*$ the anti-crossing between the two lowest eigenvalues. We plot the evolution of the first three eigenvalues for the two cases of interest. On figure \ref{Ei_toy1} (a), two types of anti-crossings are distinguishable: the one between $E_2$ and $E_1$ which is quite weak and the one between $E_1$ and $E_0$ which is strong and of interest in AQC. Here, the slope of $E_0(s)$ before $s^*$ is the slope of $E_1(s)$ after $s^*$. In other words, the lowest energy was going toward $E_1(1)$ before bouncing against the second lowest energy to redirect toward $E_0(1)$. If we follow the intuition given in the previous section, we can expect to see $a_0$ crossing $a_1$. Now, looking at figure \ref{Ei_toy1} (b), we see that $E_2$ and $E_1$ are getting closer before and even more after the anti-crossing between $E_1$ and $E_0$. The slope of $E_0$ before $s^*$ is jumping from one level to the other ending in $E_2(1)$. Our intuition says that $a_2$ is becoming dominant before $s^*$ and crosses with $a_0$. The latter case doesn't follow the parametrization definition \ref{ChoiAC}. The plots of $a_k$'s and $b_k$'s for these two instances on figure \ref{AkBk_toy1} validate our expectations on the behavior of these curves. On the left, figures \ref{AkBk_toy1}(a) and (c) show the evolution described by definition \ref{ChoiAC}, however, figures \ref{AkBk_toy1}(b) and (d) show other variations of the quantities $a_k$ and $b_k$ for an anti-crossing. Therefore, we have created an instance that doesn't follow the definition introduced by V. Choi and explain precisely why her parametrization cannot be satisfied in general. \\

\begin{figure}[ht]
\centering
$$
\begin{array}{cc}
\includegraphics[scale=0.3]{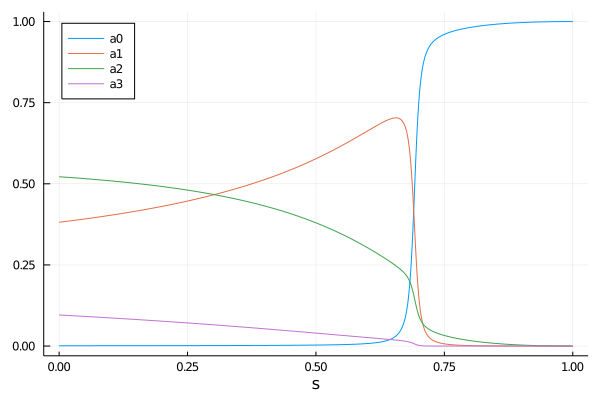} &
\includegraphics[scale=0.3]{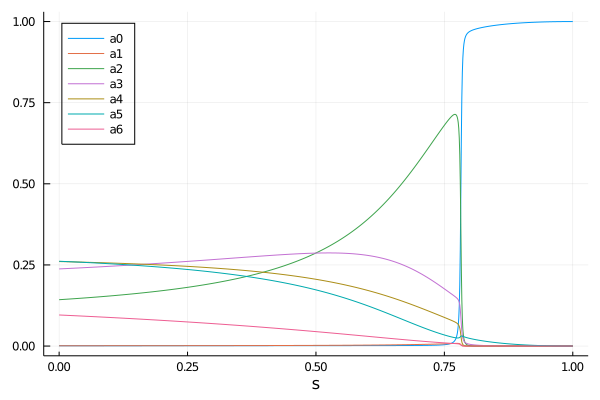} \\
(a) & (b) \\
\includegraphics[scale=0.3]{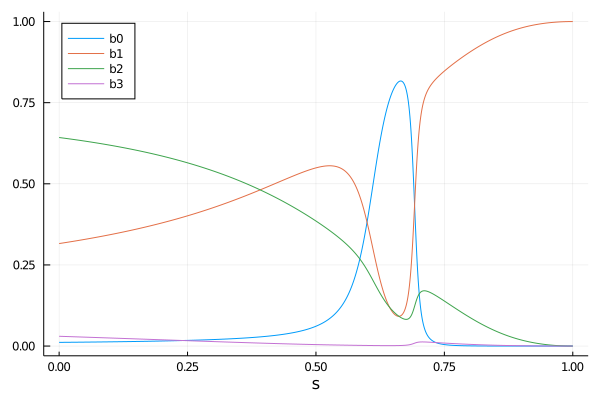} &
\includegraphics[scale=0.3]{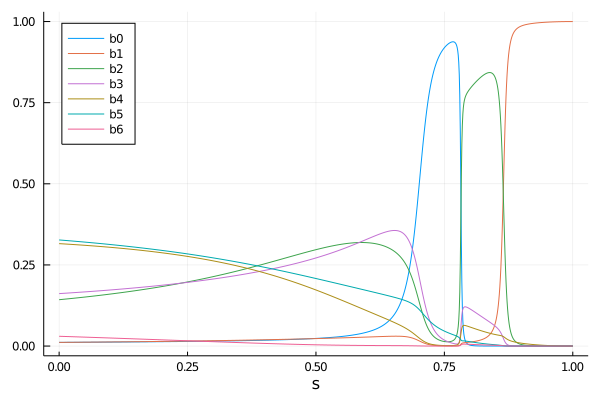} \\
(c) & (d)
\end{array}
$$
\caption{$a_k$ (top) and $b_k$ (bottom) during evolution for toy example 1 with $\alpha=0$ (left) and $\alpha=0.5$ (right)}
\label{AkBk_toy1}
\end{figure}

\noindent Here we gave more insights on the quantities introduced by V. Choi to describe anti-crossing and we explicit the link with the min-gap. A strong anti-crossing will end up in a small min-gap. However, her definition isn't general enough as the $a_k$'s and $b_k$'s don't explain exactly the behavior described by the definition. In the next section, still motivated by Choi's article, we give a new parametrization definition of an anti-crossing between the two lowest energies based only on $a_0(s)$ and $b_0(s)$ and derive a result that quantifies the min-gap. 

\section{New characterization of anti-crossing}

In this section, we give a new general result on the derivatives of the instantaneous vectors involved in the anti-crossing and then we present a new characterization of an anti-crossing more general. 

\subsection{Variation of $|E_0(s)\rangle$ at $s^*$ }
During AQC, we are interested in the evolution of the ground-state as it encodes the solution of the given problem at the end of the computation. We know that it might undergo a physical phenomenon called anti-crossing that leads to an exponentially small min-gap. By adiabatic theorem, the state stays close to the instantaneous ground-state. Why is it hard to follow it sometimes? What is happening to the instantaneous ground-state and first excited state at crossing point? We give one element of answers in the next theorem, proving that the variation of the two states that crossed are symmetrically similar. The amplitude of this variation is inversely proportional to the min-gap, the direction of the variation is only supported by the other vector state involved in the anti-crossing.

\begin{theorem}
\label{ACtheorem}
At anti-crossing point $s^*$ between $E_0(s)$ and $E_1(s)$ the following is true:
\begin{align}
\frac{d}{ds}|E_0(s^*) \rangle = &-\beta |E_1(s^*) \rangle \\
\frac{d}{ds}|E_1(s^*) \rangle = & \quad \  \beta |E_0(s^*) \rangle 
\end{align}where $\beta = \frac{\langle E_0(s^*) |\dot{H}|E_1(s^*)\rangle}{\Delta_{min}}$ 

\end{theorem}
\begin{proof}
Let's write the expression we get using the second formula of lemma \ref{eigenrelation}, for $i=0,1$ :
$$
\frac{d}{ds}|E_i(s) \rangle = (-1)^{i+1}\frac{\langle E_0(s) |\dot{H}|E_1(s) \rangle}{E_1(s) - E_0(s)} |E_{i \oplus 1}(s)\rangle - \sum_{j \geq 2} \frac{\langle E_j(s) |\dot{H}|E_i(s) \rangle}{E_j(s) - E_i(s)} |E_j(s)\rangle
$$We recognize $\beta$ as the first coefficient and a sum over levels higher than 2. Let's show that the sum is null. To do so, we will show that for $i=0,1$,$\forall j \geq 2, \langle E_i(s^*) |\dot{H}|E_j(s^*) \rangle=0$. We start from the definition of an anti-crossing by Wilkinson, where the two curves $E_0(s)$ and $E_1(s)$ behave like two hyperbolas around $s^*$ and therefore they satisfy $\frac{dE_0}{ds}(s^*-\delta) = \frac{dE_1}{ds}(s^*+\delta)$ for a small $\delta > 0$. This means that the slope is transitioning from level 0 to level 1. Then we use Taylor expansion at the first order in $\delta$ :
$$
\frac{dE_0}{ds}(s^*) - \delta  \frac{d^2E_0}{ds^2}(s^*) = \frac{dE_1}{ds}(s^*)+ \delta  \frac{d^2E_1}{ds^2}(s^*)
$$In the neighborhood of $s^*$, the gap $\Delta(s)$ is minimal in $s^*$, thus $\frac{dE_1}{ds}(s^*) - \frac{dE_0}{ds}(s^*)=\frac{d\Delta}{ds}(s^*)  = 0$. We are left with :
$$
\frac{d^2E_0}{ds^2}(s^*) +\frac{d^2E_1}{ds^2}(s^*) =0
$$We are in a setting with linear interpolation, so $\ddot{H}=0$. We can use the third expression of lemma \ref{eigenrelation} to get:
$$
2\sum_{j \neq 0} \frac{\langle E_0(s^*) |\dot{H}|E_j (s^*)\rangle ^2}{E_0(s^*) - E_j(s^*)}  + 2\sum_{j \neq 1} \frac{\langle E_1(s^*) |\dot{H}|E_j (s^*)\rangle ^2}{E_1(s^*) - E_j(s^*)} = 0
$$Pulling out from the first sum the term $j=1$ and the term $j=0$ from the second one, they cancel each other and we end up with:
$$
\sum_{j \geq 2} \left [\frac{\langle E_0(s^*) |\dot{H}|E_j (s^*)\rangle ^2}{ E_j(s^*)-E_0(s^*) }  +\frac{\langle E_1(s^*) |\dot{H}|E_j (s^*)\rangle ^2}{E_j(s^*)-E_1(s^*) } \right]= 0
$$The sum of positive summands is equal to zero, so each summand is equal to zero. Because the denominators are strictly positives, we obtain:
$$
\forall j \geq 2, \left \{ \begin{array}{ll}
\langle E_0(s^*) |\dot{H}|E_j (s^*)\rangle =0  \\
\langle E_1(s^*) |\dot{H}|E_j (s^*)\rangle =0
\end{array} \right.
$$This is enough to conclude. 
\qed
\end{proof}
\noindent Actually, with some manipulations of the eigen relation like in the proof of proposition 1,  we have that for $i \neq j$, $\langle E_i(s) |\dot{H}|E_j(s)\rangle = \frac{1}{s} \langle E_i(s)|-H_0|E_j(s)\rangle =\frac{1}{1-s}\langle E_i(s)|H_1|E_j(s)\rangle \geq 0$.  Thus, one can show that:
$$
\forall j \geq 2, \left \{ \begin{array}{ll}
\langle E_0(s^*) |H_1|E_j (s^*)\rangle =\langle E_0(s^*) |H_0|E_j (s^*)\rangle =0  \\
\langle E_1(s^*) |H_1|E_j (s^*)\rangle =\langle E_1(s^*) |H_0|E_j (s^*)\rangle =0
\end{array} \right.
$$Furthermore, if $H_1$ is positive semi-definite (which is easy to make it like that), we have $\beta \geq 0$ . 
\\ \\
Here, we understand that the smaller the min-gap is, the more brutal the variation of the instantaneous ground-state will be. This phenomenon can explain why it is hard to follow the trajectory of the ground-sate vector. Furthermore, the direction of the variation is purely supported by the other vector involved in the anti-crossing meaning that $|E_0(s^*)\rangle$ and $\frac{d}{ds}|E_0(s^*)\rangle$ form the same plane as $|E_0(s^*)\rangle$ and $|E_1(s^*)\rangle$ where these two vectors rotate. 

\subsection{New anti-crossing definition}

During adiabatic evolution, each level of energy will undergo some anti-crossings. We aim to explain the anti-crossings happening to one specific level namely the ground-state. In AQC, the purpose is to find the ground-state of the final Hamiltonian, thus knowing from which higher energy level $|GS\rangle$ comes, will help to understand if the solution has to jump many levels before ending in the ground-state. This information is present in the element of $|E_i(s)\rangle$ corresponding to $|GS\rangle$. Therefore, instead of looking at the decomposition of the instantaneous ground-state $|E_0(s)\rangle$ in the computational basis, we can look at the decomposition of the final ground-state $|GS\rangle$ in the instantaneous basis.

\begin{align}
g_k(s)=|\langle GS|E_k(s)\rangle|^2
\end{align}The value of the $g_k$ during the evolution corresponds to the probability to measure the solution if the system is in the $k^{th}$ energy level. In AQC, generally, we start from the ground-state of the initial Hamiltonian $H_0$, so the state stays close to the instantaneous ground-state according to the adiabatic theorem. Therefore, we hope that $g_0$ is dominant at some point. \\

\noindent Notice that $g_0=a_0$ and $g_1=b_0$. The idea is to relax definition \ref{ChoiAC} to include a better description of an anti-crossing between $E_0$ and $E_1$. The intuition behind those variables are quite similar, a dominant $g_j$ during the evolution means that the instantaneous basis vector $|E_j(s)\rangle$ is in direction toward the solution before $s^*$. We can restate the parametrization of an $(\gamma, \epsilon)$-anti-crossing based only on $g_0$ and $g_1$.

\begin{definition}
\label{gAC}
For $\gamma \geq 0$, $\epsilon \geq 0$ we say there is an $(\gamma, \epsilon)$-Anti-crossing if there exists a $\delta>0$ such that
\begin{enumerate}
\item For $s\in[s^*-\delta,s^*+\delta]$, \begin{align}
|GS \rangle=g_0(s)|E_0(s) \rangle + g_1(s) |E_1(s) \rangle  
\end{align} where $g_0(s)+g_1(s) \in [1-\gamma,1]$. Within the time interval $[s^*-\delta, s^*+\delta]$, $|GS\rangle$ is mainly composed of $|E_0(s)\rangle$ and $E_1(s)\rangle$. That is, all other states (eigenstates of the Hamiltonian $H(s)$) are negligible (which sums up to at most $\gamma \geq 0$).
\item  At the avoided crossing point $s=s^*$, $g_0=g_1 \in [1/2 - \epsilon,1/2]$, for a small $\epsilon >0$. That is, $|GS\rangle\simeq 1/\sqrt{2}(E_0(s^*)  + E_1(s^*) )$.
\item  Within the time interval $[s^*-\delta, s^*+\delta]$, $g_0(s)$ increases from $\leq \gamma$ to $\geq (1-\gamma)$, while $g_1(s)$ decreases from $\geq (1-\gamma)$ to  $\leq \gamma$. \\
\end{enumerate}
\end{definition}

\noindent Our new definition is quite similar to V Choi's one and trivially includes all anti-crossings described by her definition. Furthermore, we have the following result that links it to the min-gap and properly characterizes the strength of an anti-crossing.

\begin{corollary}
\begin{align}
\frac{dg_0}{ds}(s^*) + \frac{dg_1}{ds}(s^*) =& 0 \\
\frac{dg_0}{ds}(s^*) - \frac{dg_1}{ds}(s^*) =& 2 g_{0,1}(s^*) \beta
\end{align}
\end{corollary}

\begin{proof}

We have $\frac{dg_i}{ds}(s)=2\langle GS|E_i(s)\rangle \langle GS|\frac{dE_i}{ds}(s) \rangle$, $g_0(s^*) = g_1(s^*)$ and theorem \ref{ACtheorem}
\qed
\end{proof}

\noindent For a strong anti-crossing, $g_0(s^*)=g_1(s^*)=g_{0,1}(s^*) \simeq 1/2$. These quantities are easy to compute and we can verify on the toy example 1 (see figure \ref{DG}) when we vary the parameter $\alpha$ from 0 to 0.66 close to the threshold 2/3 that we observe what the latter corollary tells us. 

\begin{figure}[ht]
\centering
$$
\begin{array}{cc}
\includegraphics[scale=0.3]{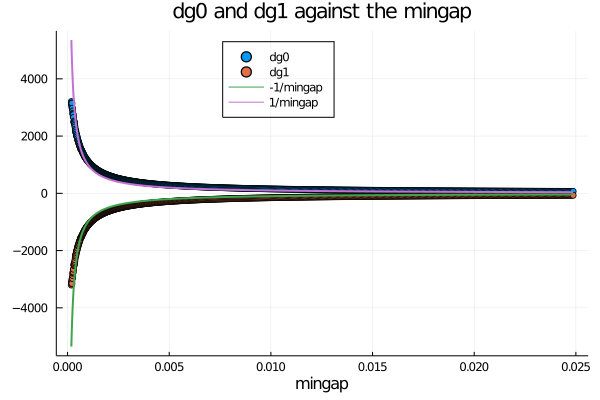} &
\includegraphics[scale=0.3]{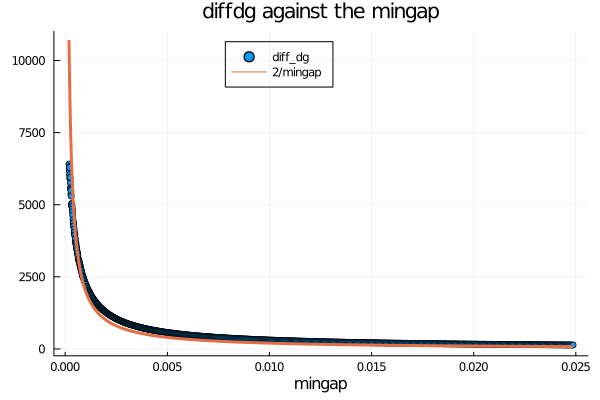} \\
(a) & (b) \\
\end{array}
$$
\caption{Evolution of the derivative of $g_0$ and $g_1$ (a) and the difference (b) against the min-gap for $\alpha$ varying from 0 (large gap) to 0.66 close to the threshold 2/3 (small gap). }
\label{DG}
\end{figure}
\noindent On figure \ref{DG}, we clearly see that the derivative of $g_0$ and $g_1$ have similar opposite variation. The plain lines express the tendency of the variations which are clearly inversely proportional to the min-gap.

\section{Future work}

Here, the results only look at anti-crossing between $E_0(s)$ and $E_1(s)$, but for further improvements in understanding the behavior of the eigen-quantities during the evolution, one could characterize an anti-crossing between $E_0(s)$ and $E_{i>1}(s)$. For example, for $i=2$, the evolution of the eigenvalues will look like figure \ref{SpecialAC}. The anti-crossing like described by Wilkinson happens between $E_0$ and $E_2$, thus we have the crossing of $g_0$ and $g_2$ but all other $g_i$'s cannot be considered negligible, especially $g_1$ which have a significant value. \\

\begin{figure}[ht]
\centering
$$
\begin{array}{cc}
\includegraphics[scale=0.3]{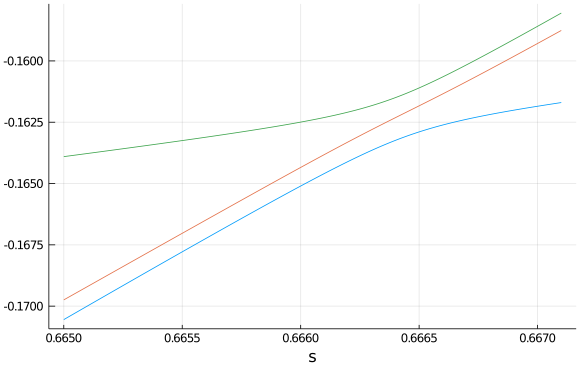} &
\includegraphics[scale=0.3]{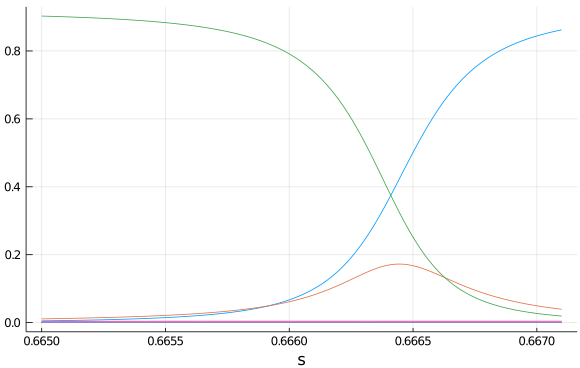} \\
(a) & (b) \\
\end{array}
$$
\caption{Special case of anti-crossing between $E_0$ and $E_2$. The energies on the left (a) and the $g_k$'s on the right (b). The range is an extra-zoom between $0.665$ and $0.667$. }
\label{SpecialAC}
\end{figure}
\noindent Another interesting track to follow is the study of starting in another initial excited state to have a high probability to measure the solution with a smaller run time. Something like that was numerically shown in \cite{Crosson} but with no clear explanation. We understand by looking the $g_k$ that it's the successive anti-crossings ending in the ground-state that advantages the initialization from a higher energy level. On figure \ref{ManyAC}, starting at the $9^{th}$ energy level would allow to stop the evolution at $58 \%$ of the total run time $T$ to have a high probability to measure the solution. This is still an open question, and it's certainly not as obvious and some hypothesis on the degeneracy of the energy levels seems necessary. 
\\ \\
\begin{figure}[ht]
\centering
$$
\begin{array}{cc}
\includegraphics[scale=0.3]{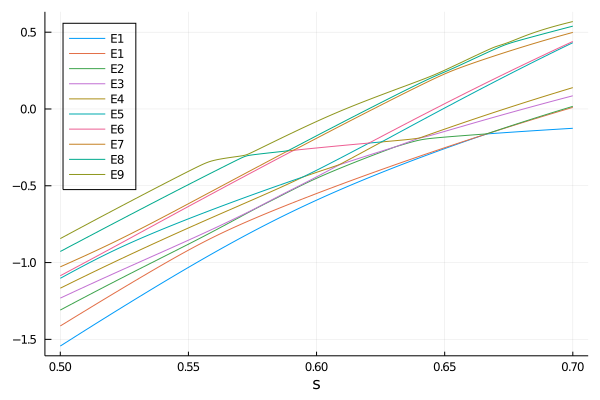} &
\includegraphics[scale=0.3]{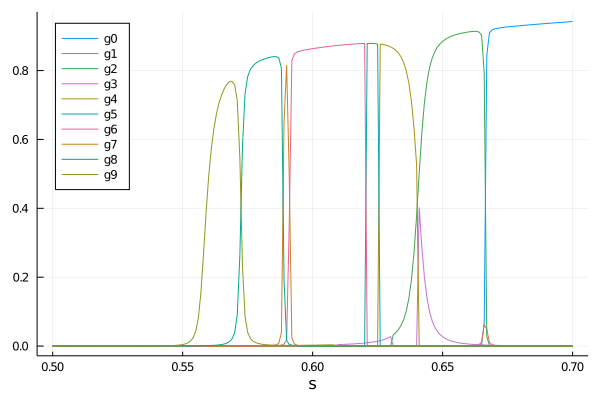} \\
(a) & (b) \\
\end{array}
$$
\caption{Successive anti-crossings ending in the ground-state. Evolution of the energy (a) and the $g_k$'s (b) between $0.52$ and $0.72$ }
\label{ManyAC}
\end{figure}

\section*{Acknowledgments}

This work was supported in part by the French National Research Agency (ANR) under the research project SoftQPRO ANR-17-CE25-0009-02, and by
the DGE of the French Ministry of Industry under the research project PIAGDN/QuantEx P163746-484124. This work has also received funding from the European Union’s Horizon
2020 research and innovation programme under grant agreement No. 817482 (PASQuanS).

\bibliographystyle{plain}
\bibliography{biblio}

\newpage
\appendix
\section{Appendices}

To illustrate the intuition on the $a_k$'s, this graph below \ref{a0a3} is the same structure of the graph figure \ref{toy1} where nodes 1 and 3 are swapped as well as nodes 5 and 6. We keep the same weights vector $w=[1,1,1,1.5,1.5,1.5]$. 

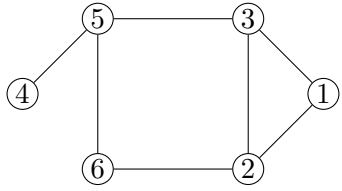
\begin{figure}[ht]
  \centering
  \begin{tikzpicture}
  \node[circle, draw, inner sep=1pt](v5) at (0,0){5};

  \node[circle, draw, inner sep=1pt](v3) at (2,0){3};
  \node[circle, draw, inner sep=1pt](v1) at (3,-1){1};
  \node[circle, draw, inner sep=1pt](v4) at (-1,-1){4};

  \node[circle, draw, inner sep=1pt](v6) at (0,-2){6};
  \node[circle, draw, inner sep=1pt](v2) at (2,-2){2};

  \draw (v3) -- (v1) -- (v2) -- (v3) -- (v5) --(v6) -- (v2) ;
  \draw (v4) -- (v5) ;

  \end{tikzpicture}
  \caption{Toy example 2}
  \label{a0a3}
\end{figure}

\noindent For $\alpha=0.2$, this produces the following final states according to their energy (\ref{StatesEi_toy2} (a)) and the eigenvalues evolution (\ref{StatesEi_toy2} (b)):

\begin{figure}[ht]
\centering
$$
\begin{array}{cc}
\includegraphics[scale=0.4]{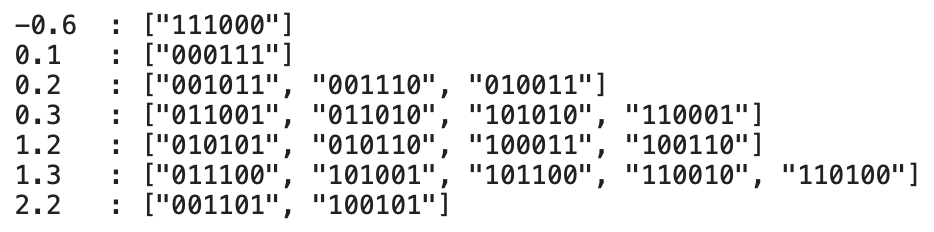} &
\includegraphics[scale=0.4]{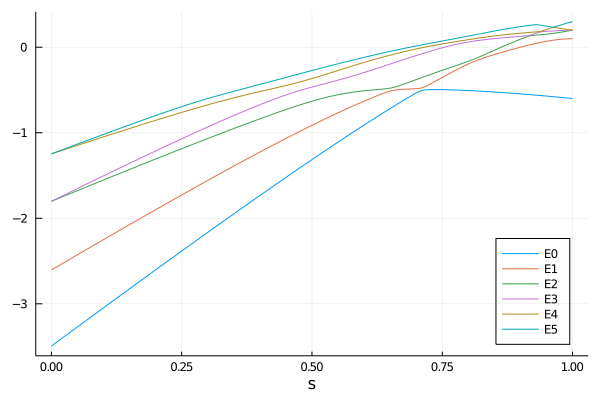} \\
(a) & (b) \\
\end{array}
$$
\caption{States (a) with their energy and $E_i(s)$ (b) during evolution for toy example \ref{a0a3} with $\alpha=0.2$}
\label{StatesEi_toy2}
\end{figure}

\noindent With this instance, we observe that the final slope of $E_0(s)$ comes from the slope of $E_2(s)$, therefore we expect that $g_2(s)$ becomes dominant before transmitting to $g_1(s)$ after the first anti-crossing between $E_2$ and $E_1$. Eventually, the anti-crossing of $E_1$ and $E_0$ will produce a crossing of $g_1$ and $g_0$. The latter will become dominant till being equal to 1 at the end. Now, focusing on the slope of $E_0(s)$ before the anti-crossing, following the different successive anti-crossings, the jumps end up in the $4^{th}$ energy level. In terms of $a_k$'s and $b_k$'s, this means that $a_3(s)$ becomes important just before the anti-crossing and crosses $a_0(s)$ at the anti-crossing. The same goes for $b_3(s)$ and $b_0(s)$. The plots below supports these previous analyses. \\ \\

\begin{figure}[ht]
\centering
$$
\begin{array}{ccc}
\includegraphics[scale=0.22]{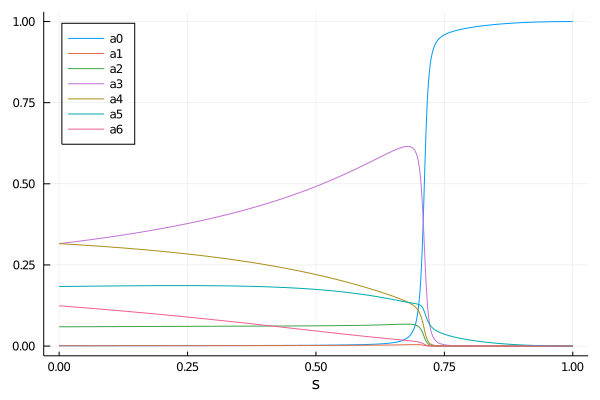} &
\includegraphics[scale=0.22]{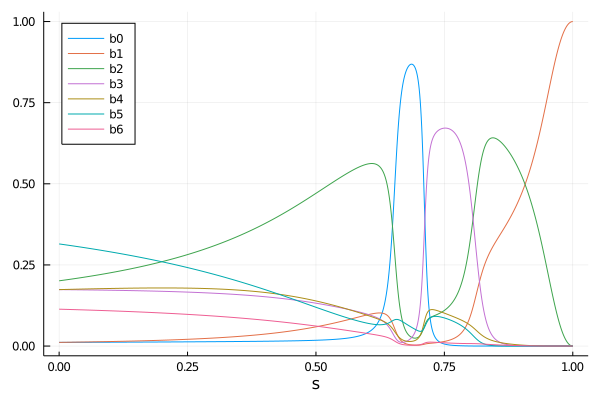}  & 
\includegraphics[scale=0.22]{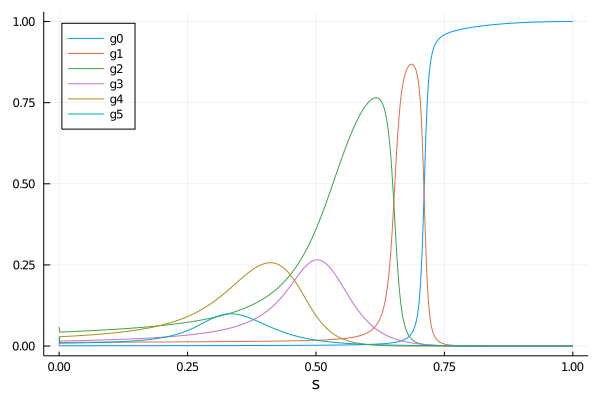} \\
(a) & (b) &(c)\\
\end{array}
$$
\caption{$a_k$ (a), $b_k$ (b) and $g_k$ (c) during evolution for toy example \ref{a0a3} with $\alpha=0.2$}
\label{AkBkGk_toy2}
\end{figure}

\noindent \textbf{Remarks}: 
\begin{enumerate}
\item On figure \ref{AkBkGk_toy2} (a), we only see one specific situation, namely the crossing of $a_3$ with $a_0$ at the anti-crossing point $s^*$ between $E_0$ and $E_1$. This means that the curve $E_0(s)$ was going toward the $4^{th}$ energy level in terms of the slope before $s^*$ and immediately change its slope toward its final direction of the $1^{st}$ energy level. Hypothetically, we could observe $a_3$ crossing $a_1$, which will indicate that $E_0(s)$ change its direction toward the $2^{nd}$ energy level, then $a_1$ crossing $a_0$. In this hypothetical case, $E_0(s)$ undergoes two anti-crossings. 
\item On figure \ref{AkBkGk_toy2} (b), we focus on the behavior of $|E_1(s)\rangle$. Here, we understand that $E_1(s)$ first went toward the $3^{rd}$ energy level before changing direction toward the lowest energy level. This is $b_2$ crossing $b_0$ when the first anti-crossing between $E_1$ and $E_2$ occurs. Then, it takes the direction of the $4^{th}$ energy level when $b_0$ crosses $b_3$. Indeed, it fetches the direction of $E_0$ before anti-crossing (remember it was $a_3$ which was dominant at this point). Then again takes back it's initial direction with $b_2$ becoming dominant at the second anti-crossing between $E_1$ and $E_2$. Eventually, smoothly change its direction toward its final goal.
\item On figure \ref{AkBkGk_toy2} (c), the point of view is quite different as we look from the final lowest energy position $E_0(1)$ and see from where it comes. We see on figure \ref{StatesEi_toy2} (b) that the final blue slope undergoes two anti-crossing before becoming blue. Indeed, it starts green, then jumps to red and finally blue. These successive anti-crossings appear on the plot of $g_k$, first, $g_2$ is dominant, then at the first anti-crossing between $E_2$ and $E_1$, $g_2$ crosses with $g_1$. Now, the final slope of $E_0$ is transported by $E_1$. Eventually, $E_1$ anti-crosses $E_0$ so $g_1$ crosses $g_0$ and the evolution (at least for the ground-state) can finish peacefully. 
\end{enumerate}

\end{document}